\documentclass[conference]{IEEEtran}
\IEEEoverridecommandlockouts
\usepackage{cite}
\usepackage{amsmath,amssymb,amsfonts}
\usepackage{graphicx}
\usepackage{textcomp}
\usepackage{tikz}
\usepackage{xcolor}
\usepackage{multirow}
\usetikzlibrary{arrows.meta,calc}

\usepackage{amsthm}

\usepackage{booktabs}
\usepackage{caption}
\usepackage{cite}
\usepackage{colortbl}
\usepackage{pdfpages}
\usepackage{dsfont}

\usepackage{bm}
\usepackage{bbm}
\usepackage{dirtytalk}

\usepackage{subcaption}

\newtheorem{theorem}{Theorem}

\newtheorem{lemma}{Lemma}

\theoremstyle{definition}

\theoremstyle{definition}
\newtheorem{remark}{Remark}

\theoremstyle{definition}

\definecolor{DarkGreen}{rgb}{0.1,0.5,0.1}
\definecolor{DarkRed}{rgb}{0.5,0.1,0.1}
\definecolor{DarkBlue}{rgb}{0.1,0.1,0.5}
\definecolor{DarkPurple}{rgb}{0.5,0.2,0.5}
\definecolor{DarkTurquoise}{rgb}{0.1,0.5,0.5}

\newcolumntype{P}[1]{>{\centering\arraybackslash}p{#1}}

\newcommand{\A}{\mathcal{A}}

\DeclareMathOperator{\supp}{supp}

\usepackage{booktabs}
\usepackage{diagbox}

\def\BibTeX{{\rm B\kern-.05em{\sc i\kern-.025em b}\kern-.08em
    T\kern-.1667em\lower.7ex\hbox{E}\kern-.125emX}}
\begin{document}

\title{Root of Unity for Secure Distributed Matrix Multiplication: Grid Partition Case\\
% {\footnotesize \textsuperscript{*}Note: Sub-titles are not captured in Xplore and
% should not be used}
% \thanks{Identify applicable funding agency here. If none, delete this.}
}

\author{\IEEEauthorblockN{Roberto Assis Machado, Felice Manganiello, \IEEEmembership{Senior~Member, IEEE}} \IEEEauthorblockA{School of Mathematical and Statistical Sciences\\ Clemson University\\Clemson, USA}}

\maketitle

\begin{abstract}
We consider the problem of secure distributed matrix multiplication (SDMM), where a user has two matrices and wishes to compute their product with the help of $N$ honest but curious servers under the security constraint that any information about either $A$ or $B$ is not leaked to any server. This paper presents a \emph{new scheme} that considers a grid product partition for matrices $A$ and $B$, which achieves an upload cost significantly lower than the existing results in the literature. Since the grid partition is a general partition that incorporates the inner and outer ones, it turns out that the communication load of the proposed scheme matches the best-known protocols for those extreme cases.
\end{abstract}
\vspace{-10pt}
% \textit{An extended version of this paper, including all proofs, is accessible at: \emph{}}

\section{Introduction}
%!TEX root = main.tex

The core and one of the most expensive operations in many machine learning applications is matrix multiplication. Performing such operations locally on a single computer takes a long time. Users would consider outsourcing their matrices to a distributed system for time-sensitive applications to carry out demanding computation tasks. Efficient methods require coding over the input matrices to speed up the computational time, yielding a trade-off among the number of workers needed, tasks performed at each worker, and the total amount of data transmitted. While outsourcing disrupts the computational delays, it leads to data security concerns. This paper aims to develop an efficient, secure distributed matrix multiplication (SDMM) scheme which keeps matrices secure from the potentially colluding servers.

 We consider the problem of secure distributed matrix multiplication (SDMM), where a user has two matrices, $A \in \mathbb{F}_q^{a \times b}$, $B \in \mathbb{F}_q^{b \times c}$ and wishes to compute their product, $AB \in \mathbb{F}_q^{a \times c}$, with the assistance of $N$ servers, without leaking any information about either $A$ or $B$ to any server. We assume that all servers are honest but curious (passive) in that they are not malicious and will follow the pre-agreed upon protocol. However, any $T$ of them may collude to eavesdrop and extrapolate information regarding $A$ or $B$.

The setting considered in this paper is proposed in~\cite{ravi2018mmult}, with many follow-up works~\cite{Kakar2019OnTC,koreans,d2019gasp,DOliveira2019DegreeTF, Aliasgari2019DistributedAP,aliasgari2020private,Kakar2019UplinkDownlinkTI,doliveira2020notes,Yu2020EntangledPC,mital2020secure,bitar2021adaptive, hasircioglu2021speeding,ftp-9606447}. The initial performance metric used was the download cost, meaning the total amount of data downloaded by the users from the server. Following papers have also considered the upload costs \cite{mital2020secure}, the total communication costs \cite{9004505, ftp-9606447}, and computational costs \cite{doliveira2020notes}.

Different partitions of the matrices lead to different trade-offs between upload and download costs. In this paper, we consider the most general one, namely, the grid product partition given by  \[A = \begin{bmatrix}
A_{i,j} 
\end{bmatrix}_{\substack{1\leq i\leq t\\ 1\leq j\leq s}}, 
B = \begin{bmatrix}
B_{i,j} 
\end{bmatrix}_{\substack{1\leq i\leq s\\ 1\leq j\leq d}}\]
such that, 
\[\displaystyle AB = 
\begin{bmatrix}
M_{1,1} & \cdots & M_{1,d}\\
\vdots & \ddots & \vdots\\
M_{t,1} & \cdots & M_{t,d}\\\end{bmatrix}=\begin{bmatrix}
\sum_{\ell=1}^s A_{i,\ell} B_{\ell,j}\end{bmatrix}_{\substack{1\leq i\leq t\\ 1\leq j\leq d}},\]
where the products $A_{i,\ell} B_{\ell, j}$ are well-defined and of the same size. Under this partition, a polynomial code is a polynomial $h(x)=f_A(x) \cdot f_B(x)$, whose coefficients encode the submatrices $A_{i, j}B_{k,\ell}$. The next step is where the scheme we propose
differs from previous works that use the grid product partition. The user evaluates polynomials $f_A(x)$ (encoding matrix $A$) and $f_B(x)$ (encoding matrix $B$) at powers $\alpha_N, \alpha^2_N, \ldots, \alpha^N_N=1$ of an $N$-th root of unity $\alpha_N$. The $N$ servers compute the product $h(\alpha^i)=f_A(\alpha^i)f_B(\alpha^i)$ for $i\in\{1,\ldots,N\}$.  The polynomial $h(x)$ is constructed so that no $T$-subset of evaluations reveals any information about $A$ or $B$ ($T$-security), but so that the user can reconstruct $AB$ given all $N$ evaluations (decodability).

An example of a polynomial scheme for the grid product partition is the secure MatDot codes in \cite{Aliasgari2019DistributedAP} and the entangled polynomial codes in~\cite{8949560}. 

In Theorem \ref{theo:scheme}, we characterize the total communication rate
achieved by our proposed scheme:

\begin{theorem}\label{theo:scheme}
Let $t,d,s$ and $T$ be positive integers. Let $A \in \mathbb{F}_q^{a \times b}$, $B \in \mathbb{F}_q^{b \times c}$ be two matrices.
Then, the proposed scheme with partition parameters $(t,d,s)$ and security parameter $T$ securely computes $AB$ with the assistance of  
\[
    N= 
\begin{cases}
    (d+1)(t+T)-1,              &\text{if }s=1 \\
    dst+dT+ts-1+T+1,& \text{if } s> 1
\end{cases}
\]

servers and a total communication rate of
\begin{align} 
\mathcal{R} = \left( N\left(\frac{b}{cts}+\frac{b}{asd}+ \frac{1}{td}\right) \right)^{-1} . 
\end{align}%\rob{needs to check and add case s=1. Also, if this is a great parameter to compare with other ones}
\end{theorem}

In Theorem~\ref{theo:comparison}, we show that our proposed code matches the recovery threshold of the best-known scheme for inner product partition, \cite{mital2020secure} and also matches GASP codes for outer product partition when $T<t$.

\begin{theorem}\label{theo:comparison}
Let $(t,d,s)$ be the partition parameters and $T$ be the security parameter:

\begin{itemize}
    \item If $t=d=1$, meaning inner product partition, then the recovery rate for the proposed scheme matches s+2T, same as scheme in \cite{mital2020secure} without pre-computation.
    \item If $s=1$, meaning outer product partition, then the recovery rate for the proposed scheme matches $(d+1)(t+T)-1$, same as for GASP codes, in \cite{9004505}, $T<t$.
\end{itemize}
\end{theorem}
\subsection{Related Work}

For distributed computations, polynomial codes were initially introduced in~\cite{polycodes1} to mitigate stragglers in distributed matrix multiplication. A series of works followed this,~\cite{polycodes2,pulkit,pulkit2,fundamental}.

The literature on SDMM has also studied different variations of the model we focus on here. For instance, in~\cite{nodehi2018limited,jia2019capacity,mital2020secure,akbari2021secure} the encoder and decoder are considered to be separate, in~\cite{nodehi2018limited} servers are allowed to cooperate. In~\cite{kim2019private} the authors consider a hybrid setup between SDMM and private information retrieval where the user has a matrix $A$ and wants to multiply it with a matrix $ B$ belonging to some public list privately.

\subsection{Main Contributions}

Our main contributions are summarized below.

\begin{itemize}
     \item We present a generalization for polynomial coding in the context of the secure distributed matrix multiplication problem, considering the grid product partition. This partition allows extending the use of techniques to reduce upload costs.
    
    \item By adapting the Fourier Discrete Transform used in \cite{mital2020secure}, we present a new scheme for SDMM. It reduces the communication load by lowering the recovery threshold. We show that they are secure, decodable, and present their total communication rate in Theorem~\ref{theo:scheme}.

    \item In Theorem~\ref{theo:comparison}, we show that \emph{the proposed scheme} matches the recovery threshold of the best-known scheme for inner product partition, \cite{mital2020secure}, and also matches the GASP scheme for outer product partition when $T<t$.

\end{itemize}

% \section{Preliminaries} \rob{since preliminaries are small, I'm thinking on incorporatign this on scheme section}

% In this section, we introduce some notation and concepts needed for the rest of the paper. We start with the definition of a $n$-th root of unity.

% \begin{definition}
% An element $\alpha_N\in \mathbb{F}_q$ is an $N$-th primitive root of unity if
% $\alpha_N^{N}=1, \mbox{ and } \alpha_N^{m}\neq 1 \mbox{ for }  0<m<N$.
% \end{definition}

% \begin{proposition}\label{decoding}
% \[
%     \sum_{i=1}^N(\alpha_N^i)^s= 
% \begin{cases}
%     0,              &\text{if }N\nmid  s \\
%     N,& \text{if } N \mid s
% \end{cases}
% .\]
% \end{proposition}

\section{A Motivating Example: $d=t=2$ and $T=1$} 
%\section{A Motivating Example: $L=2$ and $T=0$} \label{sec: motivating example}
\label{sec2}

We begin the description of our proposed scheme with an example, which we present to showcase our scheme. At the end of the section, we assume that each server can compute $\frac{abc}{4}$ scalar operations, meaning additions or multiplications in $\mathbb{F}_q$. Finally, we compare the proposed method with GASP codes that use an outer partition product and the one using an inner partition product.

In this example, a user wishes to multiply two matrices $A \in \mathbb{F}_{q}^{a \times b}$ and $B \in \mathbb{F}_{q}^{b \times c}$ with the assistance of non-colluding helper servers. Consider the following matrices are partitioned as follows
\[A=\begin{bmatrix}
A_{1,1} & A_{1,2}\\A_{2,1} & A_{2,2}
\end{bmatrix}\in \mathbb{F}_q^{a\times b}, \ B=\begin{bmatrix}
B_{1,1} & B_{1,2}\\B_{2,1} & B_{2,2}
\end{bmatrix}\in \mathbb{F}_q^{b\times c}\]
By multiplying the matrices we obtain
\[M=AB=\begin{bmatrix}
A_{1,1}B_{1,1}+A_{1,2}B_{2,1} & A_{1,1}B_{1,2}+A_{1,2}B_{2,2}\\A_{2,1}B_{1,1}+A_{2,2}B_{2,1}&A_{2,1}B_{1,2}+A_{2,2}B_{2,2}.
\end{bmatrix}\]

Since we assume non-colluding servers, i.e, $T=1$, it involves picking two random matrices $R \in \mathbb{F}_{q}^{\frac{a}{2} \times \frac{b}{2}}$ and $S \in \mathbb{F}_{q}^{\frac{b}{2} \times \frac{c}{2}}$.
Consider the (Laurent) polynomials
\begin{align*}
    f_A(x)&=A_{1,1} + A_{1,2}x + A_{2,1}x^2+ A_{2,2}x^3 + R x^4,\\ &\mbox{ and } \\
    f_B(x)&=B_{1,1} + B_{2,1}x^{-1} + B_{1,2}x^{-5}+ B_{2,2}x^{-6} + S x^{-10}
\end{align*}

% \textcolor{red}{Looking at matrix $M$ and at the multiplication of the two polynomials, we need the following conditions to hold
% \[\left\{\begin{array}{l}
%      \alpha_1+\beta_1=\alpha_2+\beta_3 \\
%       \alpha_1+\beta_2=\alpha_2+\beta_4 \\
%       \alpha_3+\beta_1=\alpha_4+\beta_3 \\
%       \alpha_3+\beta_2=\alpha_4+\beta_4 
% \end{array}\right.\]
% which is equivalent to 
% \[\alpha_1-\alpha_2=\alpha_3-\alpha_4=\beta_3-\beta_1=\beta_4-\beta_2.\]
% }

We obtain the following degree table for polynomial $h(x) = f_A(x)f_B(x)$:
\[
\begin{array}{c|ccccc}
+ & 0 & 1 & 2 &3&  4 \\
\hline
 0& {\color{blue}0 }& 1 & {\color{green}2} & 3 & 4 \\
-1 & -1 & {\color{blue}0} & 1 & {\color{green}2} & 3  \\
-5 & {\color{red}-5} & -4 & {\color{orange}-3} & -2& -1  \\
-6 & -6 & {\color{red} -5} & -4 & \color{orange}{-3}  &-2 \\
-10 & -10 & -9 & -8 & -7 & -6 \\\end{array}
\]
leading to a problem to find evaluations points $\mathbb{F}_q$ that minimizes the set
$\{\alpha^{i}: i\in \{-10, -9, \ldots, 4\} \}$ under the following conditions:

\begin{itemize}
    \item $|\{\alpha^{-5}, \alpha^{-3}, \alpha^{0}, \alpha^{2}\}|=4$
    \item $\{\alpha^{-5}, \alpha^{-3}, \alpha^{0}, \alpha^{2}\} \cap \{\alpha^{-10},\ldots, \alpha^{-6}, \alpha^{-4}, \\ \mbox{\hspace{4.5cm}}\alpha^{-2}, \alpha^{-1}, \alpha,  \alpha^{3}, \alpha^{4}\}=\emptyset$
    
\end{itemize}

Consequently, if $\alpha=\alpha_{13}$ is an $13$-th root of unity, such a condition is satisfied.

\subsection{Computational complexity}

Let $\alpha_{13}\in \mathbb{F}_q$ be a primitive root of unity. The algorithm for computing the multiplication is as follows
\begin{enumerate}
    \item \textbf{Encode.} For $i=1,\ldots,13$, the user computes $f_A(\alpha_{13}^i)$ and $f_B(\alpha_{13}^i)$.
    \item \textbf{Upload.} The user sends matrices $f_A(\alpha_{13}^i)$ and $f_B(\alpha_{13}^i)$ to Server $i$. 
    \item \textbf{Server multiplication.} Servers multiply together the received matrices.
    \item \textbf{Download.} Servers send the result $f_A(\alpha_{13}^i)\cdot f_B(\alpha_{13}^i)$ back to the user. 
    \item \textbf{Decode.} The user uses Equation \ref{prop:rootofunity} to obtain the coefficients with degree $-5$, $-3$, $0$, and $2$ or, equivalently, $0$, $2$, $8$, and $10$ since polynomials are evaluated at an 13-th root of unity $\alpha_{13}$. Therefore,
    
    \begin{align*}
        A_{1,1}B_{1,1}+A_{1,2}B_{2,1}&=\frac{1}{13}\sum_{i=1}^{13}f_A(\alpha_{13}^i)\cdot f_B(\alpha_{13}^i)\\ 
        A_{1,1}B_{1,2}+A_{1,2}B_{2,2}&=\frac{1}{13}\sum_{i=1}^{13}(\alpha_{13}^i)^5f_A(\alpha_{31}^i)\cdot f_B(\alpha_{13}^i)\\
        A_{2,1}B_{1,1}+A_{2,2}B_{2,1}&=\frac{1}{13}\sum_{i=1}^{13}(\alpha_{13}^i)^{11}f_A(\alpha_{13}^i)\cdot f_B(\alpha_{13}^i)\\
        A_{2,1}B_{1,2}+A_{2,2}B_{2,2}&=\frac{1}{13}\sum_{i=1}^{13}(\alpha_{13}^i)^3f_A(\alpha_{13}^i)\cdot f_B(\alpha_{13}^i)
    \end{align*}
\end{enumerate}

We start with the assumption that addition and multiplication in $\mathbb{F}_q$ take constant time. We consider for simplicity that parameters $a$, $b$, and $c$ are divisible by $2$. We describe below here the complexities of each step:

\begin{enumerate}
    \item Computing $f_A(\alpha_{13}^i)$ and $f_B(\alpha_{13}^i)$ requires $2ab$ and $2bc$  $\mathbb{F}_q$-operations, respectively. This translates to $26(ab+bc)$ $\mathbb{F}_q$-operations to compute the $13$ evaluations.
    \item The user sends $\frac{13}{4}(ab+bc)$ $\mathbb{F}_q$-elements to the servers.
    \item The computational cost to perform the product $f_A(\alpha_{13}^i)f_B(\alpha_{13}^i)$ on each server is $\frac{ac(b-1)}{4}$.
    \item Each server sends $\frac{ac}4$ $\mathbb{F}_q$-elements to the user.
    \item The decoding step requires up to $\frac{13ac}{2}$ $\mathbb{F}_q$-operations to obtain each coefficient of interest. Since $3$ of those requires exactly $\frac{13ac}{2}$ and one requires $\frac{14ac}{4}$, then in total, we need $23ac$ $\mathbb{F}_q$-operations are required to retrieve the desired product $AB$.
\end{enumerate}

\begin{table*}[ht]
\centering
\begin{tabular}{|P{3cm}||P{3cm}|P{3cm}|P{3cm}|P{3cm}|}
%  \hline
%  \multicolumn{5}{|c|}{Comparison to other methods} \\
%  \hline
\textbf{Scheme} & \textbf{Upload Cost} & \textbf{Download Cost} & \textbf{Encoding Complexity} & \textbf{Decoding Complexity}\\
 \hline
\rule[-2ex]{0pt}{6ex} Proposed Scheme   &  $\frac{13}{4}(ab+bc)$   & $\frac{13 ac}{4}$& $26(ab+bc)$ & $23ac$ \\ \hline
\rule[-2ex]{0pt}{6ex}  GASP& $\frac{7}{2}(ab+bc)$ & $\frac{7ac}{4}$&$28(ab+bc)$&$27ac$\\ \hline
\rule[-2ex]{0pt}{6ex}  Scheme in \cite{mital2020secure} & $\frac{3}{2}(ab+bc)$ & $7ac$& $12(ab+bc)$ & $7ac$\\ \hline
\end{tabular}
\caption{Comparison to other methods with limited $\mathbb{F}_q$-operations   $\frac{ac(b-1)}{4}$ to compute $f_A(\alpha_{11}^i)f_B(\alpha_{11}^i)$ on each server.}{\label{tab:compare}}
\end{table*}
\begin{remark}
    If we consider the time to transmit one, add or multiply two elements in $\mathbb{F}_q$ is equal to $1$, the scheme presented can speedup computational time of multiplying matrices $A \in \mathbb{F}_{q}^{a \times b}$ and $B \in \mathbb{F}_{q}^{b \times c}$ if the dimensions of the matrices satisfy $a > \frac{234}{7}$, $b > \frac{216 a}{-234 + 7 a}$, and $c > \frac{234 a b}{-216 a - 234 b + 7 a b}$ compared to local computation which requires $2abc-ac$ operations using the traditional matrix multiplication.

%a>234/7, b>(216 a)/(7 a - 234), c>(234 a b)/(7 a b - 216 a - 234 b)
\end{remark}
%The overall computation cost is $12(ab+bc)+\frac{abc}{16}+\frac{15}4ac$ and the communication cost is $2(ab+bc+ac)$. 

In Table \ref{tab:compare}, we present a comparative summary for this example among the proposed method, GASP (which uses outer product partition), and the scheme shown in \cite{mital2020secure} (which inner product partition). For this comparison, we fixed the amount of $\mathbb{F}_q$-operations in each server by $\frac{ac(b-1)}{4}$; therefore, we shall assume parameters $a$, $b$ and $c$ are divisible by $4$.

%The computation cost of GASP is $4(ab+bc)+\frac{abc}4+\frac{7}2ac$
\begin{remark}\label{rem:fieldsize}
Since the evaluation points are powers of an $N$-th primitive root of unity, the appropriate size $q$ of the field should satisfy $N \mid (q - 1)$.  This condition ensures the existence of the multiplicative inverse of $N$ in $\mathbb{F}_q$ so that decodability is guaranteed.
\end{remark}
\section{Proposed Scheme}

This section is devoted to presenting the general construction of the proposed scheme. We perform the same technique as in Section~\ref{sec2} retrieving the $dt$ matrices $\sum_{\ell=1}^sA_{i,\ell} B_{\ell,j}$ from the polynomial $h(x) = f_A(x) \cdot f_B(x)$. 

% Consider the following matrices $A \in \mathbb{F}_q^{a \times b}$, $B \in \mathbb{F}_q^{b \times c}$ partitioned as
% $A = \begin{bmatrix}
% A_{i,j} 
% \end{bmatrix}_{\substack{1\leq i\leq t\\ 1\leq j\leq s}}$
% $B = \begin{bmatrix}
% B_{i,j} 
% \end{bmatrix}_{\substack{1\leq i\leq s\\ 1\leq j\leq d}}$, such that  \[\displaystyle AB =\begin{bmatrix}
% \sum_{\ell=1}^s A_{i,\ell} B_{\ell,j}\end{bmatrix}_{\substack{1\leq i\leq t\\ 1\leq j\leq d}}\]

% Consider polynomials \begin{align*}f_A(x) = & \sum_{i=1}^t\sum_{j=1}^s A_{i,j}x^{(i-1)s+j-1} + \sum_{k=1}^{T} R_kx^{ts+k-1},\\ f_B(x) = & \sum_{i=1}^s\sum_{j=1}^d B_{i,j}x^{(1-j)(ts+T)+(1-i)} + \sum_{k=1}^{T} S_kx^{(-d)(ts+T)-k+1}.
% \end{align*}

\bigskip

\noindent \textbf{Choosing the Polynomials:} As described in the introduction, the user partitions the matrices $A \in \mathbb{F}_q^{a \times b}$ and $B \in \mathbb{F}_q^{b \times c}$ as $A = \begin{bmatrix}
A_{i,j} 
\end{bmatrix}_{\substack{1\leq i\leq t\\ 1\leq j\leq s}}$
$B = \begin{bmatrix}
B_{i,j} 
\end{bmatrix}_{\substack{1\leq i\leq s\\ 1\leq j\leq d}}$ with the purpose of getting the matrix multiplication expressed as  \[\displaystyle AB =\begin{bmatrix}
\sum_{\ell=1}^s A_{i,\ell} B_{\ell,j}\end{bmatrix}_{\substack{1\leq i\leq t\\ 1\leq j\leq d}},\] where  $A_{i,j}\in\mathbb{F}_q^{\frac{a}{t} \times \frac{b}{s}}$ and $B_{i,j}\in\mathbb{F}_q^{\frac{b}{s}\times \frac{c}{t}}$. To obtain $T$-security $R_1, \ldots, R_T \in \mathbb{F}_q^{\frac{a}{t} \times \frac{b}{s}}$ and $S_1, \ldots, S_T \in \mathbb{F}_q^{\frac{b}{s}\times \frac{c}{t}}$ are chosen independently and uniformly at random. We then define $f_A\in \mathbb{F}_q^{\frac{a}{t} \times \frac{b}{s}}[x,x^{-1}]$ and $f_B\in \mathbb{F}_q^{\frac{b}{s}\times \frac{c}{t}}[x,x^{-1}]$ as the following polynomials

\begin{equation}\label{eq:encfunction}
 f_A(x) =  \sum_{i=1}^t\sum_{j=1}^s A_{i,j}x^{(i-1)s+j-1} + \sum_{k=1}^{T} R_kx^{ts+k-1},
 \end{equation}
\begin{multline*}
   f_B(x) = \sum_{i=1}^s\sum_{j=1}^d B_{i,j}x^{(1-j)(ts+T)+(1-i)}+ \\+\sum_{k=1}^{T} S_kx^{(-d)(ts+T)-k+1}.
\end{multline*}

\noindent \textbf{Choosing the Field and Evaluation Points:}  Let $J\subset \mathbb{Z}$ be a finite set and $p(x)=\sum_{i\in J}M_ix^i\in \mathbb{F}_q^{m_1\times m_2}[x,x^{-1}]$. Define the support set of $p(x)$ to be  \[\supp (p)= \{i\in J: M_i\neq 0\}.\] To choose the evaluation points in $\mathbb{F}_q$, we need to look for the $N$-th primitive root of unity $\alpha_N$ that minimizes the set of exponents
%$\mathcal{S}(i,j,k_1,k_2)=\{(i-1)s + (1-j)(ts+T) + k_1 - k_2: 1\leq i \leq t, 1\leq j\leq d \mbox{ and } 1\leq k_1, k_2\leq s\}$ 
$\supp (f_A) + \supp (f_B) = \supp (h)$ under the following conditions:

\begin{itemize}
    \item $|\{\alpha_N^i: i\in \supp (f_A)\}| = ts + T$
    \item $|\{\alpha_N^i: i\in \supp (f_B)\}| = ds + T$
    \item $|\mathcal{I}|=|\{\alpha_N^{(i-1)s + (1-j)(ts+T)}: 1\leq i\leq t, 1\leq j \leq d\}|=td$
    \item  $\alpha_N^{z}\notin \mathcal{I}$, for any power $z$ of polynomial $h(x)$ associated to coefficients  $A_{i,k_1}B_{k_2,j}$ with $k_1\neq k_2$, and any  coefficient multiple of $R_k$ or $S_k$.
    % and $A_{i,j}S_k$, $R_kB_{i,j}$ and $R_{j_1}S_{j_2}$, for $i, j, k, j_1, j_2, k_1, k_2$ in their proper intervals with . %Coefficients of polyno associated to %$\mathcal{S}(i,j,k,k)\cap \mathcal{S}(i,j,k_1,k_2) = \emptyset$, for every $k, k_1, k_2\in\{1, 2, \ldots, s\}$, with $k_1\neq k_2$.
\end{itemize}

If $s=1$, then a $N=((d+1)(t+T)-1)$-th primitive root of unity satisfies the conditions. Otherwise, if $s>1$, a $N=(dst+dT+ts-1+T+1)$-th primitive root of unity will do so. 

Therefore, Remark \ref{rem:fieldsize} establishes the following condition on the size $q$ of the finite field $ ((d+1)(t+T)-1)\mid q$  if $s=1$, or $(dst+dT+ts-1+T+1)\mid q$, otherwise.

\noindent \textbf{Upload Phase:} The proposed scheme uses $N$ servers, as determined in the previous item. The user uploads $f_A(\alpha^i_N)$ and $f_B(\alpha^i_N)$ to each Server $i$.

\noindent \textbf{Download Phase:} The $i$-th server computes the matrix multiplication $f_A(\alpha^i_N)\cdot f_B(\alpha^i_N)$ and sends its result back to the user.

\noindent \textbf{User Decoding:} In Lemma~\ref{lem:decodability}, we show that the user is able to retrieve $\sum_{\ell=1}^sA_{i,\ell} B_{\ell,j}$ from 
$\{h(\alpha^i_N): i\in\{1, \dots, N\}\}$. Combining these, the user can decode \[\displaystyle AB =\begin{bmatrix}
\sum_{\ell=1}^s A_{i,\ell} B_{\ell,j}\end{bmatrix}_{\substack{1\leq i\leq t\\ 1\leq j\leq d}}.\]

\section{Proof of Theorem~\ref{theo:scheme}}
We break the proof into different Lemmas. We show that the proposed scheme is decodable, in Lemma~\ref{lem:decodability}, $T$-secure, in Lemma~\ref{lem:tsecure}, and characterize their performance, in Lemma~\ref{lem:comcosts}. These statements combined prove Theorem~\ref{theo:scheme}. 

\begin{lemma}\label{lem:decodability}
 Let $A \in \mathbb{F}_q^{a \times b}$, $B \in \mathbb{F}_q^{b \times c}$ be two matrices. Given positive integers $t,d,s$ and $T$ let $(t,d,s)$ be the partition parameters, meaning \[A = \begin{bmatrix}
A_{i,j} 
\end{bmatrix}_{\substack{1\leq i\leq t\\ 1\leq j\leq s}}, 
B = \begin{bmatrix}
B_{i,j} 
\end{bmatrix}_{\substack{1\leq i\leq s\\ 1\leq j\leq d}}.\]
Then, $\sum_{k=1}^sA_{i,k}B_{k,j}$ can be decoded using $N$ servers, for $1\leq i\leq t$ and $1\leq j\leq d$.
\end{lemma}

\begin{proof}

Let $f_A(x)\in \mathbb{F}_q^{\frac{a}{d} \times \frac{b}{s}}[x,x^{-1}]$ and  $f_B(x)\in \mathbb{F}_q^{\frac{b}{s} \times \frac{c}{t}}[x,x^{-1}]$ be polynomials defined by
\begin{align*}f_A(x) = & \sum_{i=1}^t\sum_{j=1}^s A_{i,j}x^{(i-1)s+j-1} + \sum_{k=1}^{T} R_kx^{ts+k-1},\\ f_B(x) = & \sum_{i=1}^s\sum_{j=1}^d B_{i,j}x^{(1-j)(ts+T)+(1-i)}+\\&\hspace{3.5cm}+ \sum_{k=1}^{T} S_kx^{(-d)(ts+T)-k+1}.
\end{align*}
using the grid product partition for matrices $A$ and $B$, and uniformly distributed random $\mathbb{F}_q$-matrices $R_i,S_i$.
Therefore, $h(x)=f_A(x)\cdot f_B(x)$ is a polynomial where the coefficient of degree $(i-1)s + (1-j)(ts+T)$ is $\sum_{\ell=1}^s A_{i,\ell} B_{\ell,j}$, for $1\leq i\leq t$ and  $1\leq j \leq d$.

Let us suppose $s=1$. Consider $\alpha_N$ to be an $N=((d+1)(t+T)-1)$-th primitive root of unity. Since we want to retrieve $\sum_{\ell=1}^s A_{i,\ell} B_{\ell,j}$, for $1\leq i\leq t$ and  $1\leq j \leq d$, we need to assure that $\alpha_N^{(i-1) + (1-j)(t+T)}$ is not equal to any $\alpha_N^{z}$, for any degree $z$ of polynomial $h(x)$ associated to coefficients containing $A_{i,j}S_k$, $R_kB_{i,j}$ and $R_{j_1}S_{j_2}$, for $i, j, k, j_1, j_2, k_1, k_2$ in their proper intervals with $k_1\neq k_2$.

We explore all the cases here:

\begin{itemize}
    \item \underline{Case $A_{i,j}S_k$:} In this first case, we need to show that  $$\alpha_N^{(i_1-1) + (1-j)(t+T)} \neq \alpha_N^{(i_2-1)-d(t+T)-k+1}\mbox{,}$$ 
    
    for $1\leq i_1, i_2 \leq t$, $1\leq j\leq d$ and $1\leq k \leq T$, which is equivalent to
    \begin{multline}\label{case1}
    %(i-1) + (1-j)(t+T)\not\equiv (-d)(t+T)-k+1 \mod N
    %(i-1) + (1-j)(t+T)+(d)(ts+T)+k-1\not\equiv 0 \mod N
    (i_1- i_2) + (d+1-j)(t+T)+k-1\equiv \\
    \equiv (i_1-i_2)-j(t+T)+k-1\not\equiv 0 \mod N
    \end{multline}
    
    Since $-N=(-d-1)(t+T)+1<(i_1-i_2)-j(t+T)+k-1<0$, then Equation \ref{case1} holds true. 
    
  %  Therefore, $(i-1) + (1-j)(t+T)+(d)(ts+T)+k-1\not\equiv 0 \mod N$ and $-j(t+T)+k\not\equiv 0 \mod N$, which is always true since $0<-j(t+T)+k< N$.
    \item \underline{Case $R_kB_{i,j}$:} In this case, we want to ensure $$\alpha_N^{(i-1)+(1-j_1)(t+T)} \neq \alpha_N^{(t+k-1)+(1-j_2)(t+T)}\mbox{,}$$ 
    for $1\leq i \leq t$, $1\leq j_1, j_2\leq d$ and $1\leq k \leq T$, which is equivalent to 
    \begin{equation}\label{case2}
        t+k+(j_1-j_2)(t+T)-i\not \equiv 0\mod N %(d+1)(t+T) - 1
    \end{equation}
    \end{itemize}
% \begin{align*}
%&(i-1)+(1-j_1)(t+T) \not \equiv (t+k-1)+(1-j_2)(t+T) \mod N    \\
% &t+k+(j_1-j_2)(t+T)-i\not \equiv 0\mod N\\
%    \iff  &t+k+(j_1-j_2)(t+T)-i\not \equiv 0\mod (d+1)(t+T) - 1\\
%    \end{align*}
    Since $0<t+k+(j_1-j_2)(t+T)-i\leq t+T+(j_1-j_2)(t+T)-i \leq d (t+T)-i < (d+1)(t+T)-1=N$, then Equation \ref{case2} holds true. 
\begin{itemize}
\item \underline{Case  $R_{j_1}S_{j_2}$:} For the last case, we need to show that 

$$\alpha_N^{(i-1) + (1-j)(t+T)} \neq \alpha_N^{(k_1-k_2)+t-d(t+T)} \mbox{,}$$ 

    for $1\leq i \leq t$, $1\leq j\leq d$ and $1\leq k_1, k_2 \leq T$, which is equivalent to
    \begin{multline}\label{case3}
    %(i-1) + (1-j)(t+T)\not\equiv (-d)(t+T)-k+1 \mod N
    %(i-1) + (1-j)(t+T)+(d)(ts+T)+k-1\not\equiv 0 \mod N
    i- 1 + (d+1-j)(t+T)+-t+k_2-k_1\equiv \\
    \equiv i-j(t+T)-t+k_2-k_1\not\equiv 0 \mod N
    \end{multline}
Since $-N=(-d-1)(t+T)+1<i-j(t+T)-t+k_2-k_1<0$, then Equation \ref{case3} holds true. 
\end{itemize}

For the case where $s>1$, we also need to consider all the cases $A_{i,k_1}B_{k_2,j}$, $A_{i,j}S_k$, $R_kB_{i,j}$ and $R_{j_1}S_{j_2}$, which follows analogously to the case $s=1$. 

Last steps assure $\alpha_N^{z}\notin \{\alpha_N^{(i-1)s + (1-j)(ts+T)}: 1\leq i\leq t, 1\leq j \leq d\}$, for any degree $z$ of polynomial $h(x)$ associated to coefficients  $A_{i,k_1}B_{k_2,j}$ with $k_1\neq k_2$, and any  coefficient multiple of $R_k$ or $S_k$.

Using the fact that 
\begin{equation}\label{prop:rootofunity}
     \sum_{i=1}^N(\alpha_N^i)^s= 
\begin{cases}
    0,              &\text{if }N\nmid  s \\
    N,& \text{if } N \mid s
\end{cases}
,\end{equation} 

for any $N$-th primitive root of unity, we ensure 
\[\sum_{\ell=1}^sA_{i,\ell} B_{\ell,j} = \frac{1}{N}\sum_{i=1}^{N}(\alpha_N^i)^{\delta_{i,j}}f_A(\alpha_{N}^i)\cdot f_B(\alpha_{N}^i)\mbox{,}\]
where $\delta_{i,j}=-(i-1)s-(1-j)(ts+T)$.

Decodability is then obtained by repeating this process for every $1 \leq i\leq t$ and $1 \leq j\leq d$:

\[\displaystyle AB =\begin{bmatrix}
\sum_{\ell=1}^s A_{i,\ell} B_{\ell,j}\end{bmatrix}_{\substack{1\leq i\leq t\\ 1\leq j\leq d}}.
\]
\end{proof}

Next, we show that the proposed scheme is $T$-secure.
\begin{lemma}\label{lem:tsecure}
The proposed scheme is $T$-secure.
\end{lemma}

% \begin{proof}
% In extended version. Follows from standard arguments.
% \end{proof}

\begin{proof}
Since $f_A(x)$ is independent from $B$ and $f_B(x)$ is independent from $A$, proving $T$-security is equivalent to showing that $I(A;f_A(\alpha_{i_1}), \ldots, f_A(\alpha_{i_T}))=I(B;f_B(\alpha_{i_1}), \ldots, f_B(\alpha_{i_T}))=0$. We prove the claim for $f_A(x)$, since the proof for $f_B(x)$ is analogous.

As defined in Equation \ref{eq:encfunction}, $f_A(x)$ is expressed as
\begin{align*}
f_A(x) = & \sum_{i=1}^t\sum_{j=1}^s A_{i,j}x^{(i-1)s+j-1} + \sum_{k=1}^{T} R_kx^{ts+k-1}.
\end{align*}
Then,
\begin{align*}
&I(A;f_A(\alpha_N^{i_1}), \ldots, f_A(\alpha_N^{i_T}))\\
=&H(f_A(\alpha_N^{i_1}), \ldots, f_A(\alpha_N^{i_T})) - H(f_A(\alpha_N^{i_1}), \ldots, f_A(\alpha_N^{i_T})|A)\\
%&\leq H(f_A(\alpha_N^{i_1})) + \cdots + H(f_A(\alpha_N^{i_T})) - H(f_A(\alpha_N^{i_1}), \ldots, f_A(\alpha_N^{i_T})|A)\\
%& = H(f_A(\alpha_N^{i_1})) + \cdots + H(f_A(\alpha_N^{i_T})) - H(f_A^{(T)}(\alpha_N^{i_1}), \ldots, f_A^{(T)}(\alpha_N^{i_T})),\\
\le & \sum_{j \in \mathcal{T}}H(f_A(\alpha_N^{j})) - H(f_A(\alpha_N^{i_1}), \ldots, f_A(\alpha_{i_T})|A)\\
=& \sum_{j \in \mathcal{T}}H(f_A(\alpha_N^{j})) - H(f_A^{(T)}(\alpha_N^{i_1}), \ldots, f_A^{(T)}(\alpha_N^{i_T})),\\
=& \frac{Tab}{st}\log(q) - H(f_A^{(T)}(\alpha_N^{i_1}), \ldots, f_A^{(T)}(\alpha_N^{i_T}))\\
\end{align*}
where $f_A^{(T)}(x) = \sum_{k=1}^{T} R_kx^{ts+k-1}$.

\bigskip

Since $\alpha_N$ is an $N$-th primitive root of unity, the evaluation points $\{\alpha_N^i: i\in \mathcal{T}\}$ are all different, and the following matrix has full rank.
\small{
\[
\left(\begin{matrix}
R_1(\alpha_N^{i_1ts})&\!R_1(\alpha_N^{i_2ts})&\cdots& R_1(\alpha_N^{i_Tts})\\
R_2(\alpha_N^{i_1(ts+1)})&\!R_2(\alpha_N^{i_2(ts+1)})&\cdots& R_2(\alpha_N^{i_T(ts+1)})\\
\vdots &\! \vdots & \ddots & \vdots\\
R_T(\alpha_N^{i_1(ts+T-1)})&\!R_T(\alpha_N^{i_2(ts+T-1)})&\cdots& R_T(\alpha_N^{i_T(ts+T-1)})\\
\end{matrix}\right)
\]}
This is because the set of $R_i's$ are linearly independent and the evaluation points are different which implies that $f_A^{(T)}(\alpha_N^{i_j})$'s are uniformly distributed in the space of the matrices $M_{\frac{a}{t}\times \frac{b}{s}}(\mathbb{F}_{q})$. Thus, 
$H(f_A^{(T)}(\alpha_N^{i_1}), \ldots, f_A^{(T)}(\alpha_N^{i_T})) = \frac{Tab}{st}\log(q)$,
and therefore, $I(A;f_A(\alpha_N^{i_1}), \ldots, f_A(\alpha_N^{i_T}))= 0$.
\end{proof}

We now characterize the total communication rate.

\begin{lemma}\label{lem:comcosts}
Proposed Scheme have total communication rate
\begin{align*} 
\mathcal{R} = \left( N\left(\frac{b}{cts}+\frac{b}{asd}+ \frac{1}{td}\right) \right)^{-1} . 
\end{align*}
\end{lemma}

% \begin{proof}
% All costs will be computed in $\mathbb{F}_{q_0}$-symbols.
% The upload and download costs can be directly calculated as
% \begin{align*}
% &\mathcal{U} = N_L\left(\frac{ab}{L}+\frac{bc}{L}\right) \prod_{j=1}^{L} p_j, \\
% &\mathcal{D} = ac\sum_{i=1}^{L} N_i \prod_{j \in [L] \setminus \{i\}} p_j
% .
% \end{align*}

% Since the matrix $AB$ has $\mathcal{S} = ac\prod_{j=1}^{L} p_j$ symbols of $\mathbb{F}_{q_0}$, it follows that the total communication rate is given by $\frac{\mathcal{S}}{\mathcal{U}+\mathcal{D}}$, which simplify to the presented formula.
% \end{proof}

%\section{Computational Complexity: comparison to other existing methods}

%\section{Decoding Complexity at the User}

%The computation complexity of the decoding complexity at the user 
%$O((P - 1)(log(P - 1))2T D/(td)))$ for Simeone \rob{need to compute the complexity to decode}.

%\section{Choosing the right parameters}

\section*{Acknowledgment}

Felice Manganiello is supported  by the NSF under grants DMS-1547399.

% The preferred spelling of the word ``acknowledgment'' in America is without 
% an ``e'' after the ``g''. Avoid the stilted expression ``one of us (R. B. 
% G.) thanks $\ldots$''. Instead, try ``R. B. G. thanks$\ldots$''. Put sponsor 
% acknowledgments in the unnumbered footnote on the first page.

\bibliographystyle{IEEEtran}
\bibliography{references.bib}

% Generated by IEEEtran.bst, version: 1.14 (2015/08/26)
\begin{thebibliography}{10}
\providecommand{\url}[1]{#1}
\csname url@samestyle\endcsname
\providecommand{\newblock}{\relax}
\providecommand{\bibinfo}[2]{#2}
\providecommand{\BIBentrySTDinterwordspacing}{\spaceskip=0pt\relax}
\providecommand{\BIBentryALTinterwordstretchfactor}{4}
\providecommand{\BIBentryALTinterwordspacing}{\spaceskip=\fontdimen2\font plus
\BIBentryALTinterwordstretchfactor\fontdimen3\font minus
  \fontdimen4\font\relax}
\providecommand{\BIBforeignlanguage}[2]{{%
\expandafter\ifx\csname l@#1\endcsname\relax
\typeout{** WARNING: IEEEtran.bst: No hyphenation pattern has been}%
\typeout{** loaded for the language `#1'. Using the pattern for}%
\typeout{** the default language instead.}%
\else
\language=\csname l@#1\endcsname
\fi
#2}}
\providecommand{\BIBdecl}{\relax}
\BIBdecl

\bibitem{ravi2018mmult}
W.-T. Chang and R.~Tandon, ``On the capacity of secure distributed matrix
  multiplication,'' in \emph{2018 IEEE Global Communications Conference
  (GLOBECOM)}, 2018, pp. 1--6.

\bibitem{Kakar2019OnTC}
J.~Kakar, S.~Ebadifar, and A.~Sezgin, ``On the capacity and
  straggler-robustness of distributed secure matrix multiplication,''
  \emph{IEEE Access}, vol.~7, pp. 45\,783--45\,799, 2019.

\bibitem{koreans}
H.~Yang and J.~Lee, ``Secure distributed computing with straggling servers
  using polynomial codes,'' \emph{IEEE Transactions on Information Forensics
  and Security}, vol.~14, no.~1, pp. 141--150, 2018.

\bibitem{d2019gasp}
R.~G.~L. D’Oliveira, S.~El~Rouayheb, and D.~Karpuk, ``Gasp codes for secure
  distributed matrix multiplication,'' in \emph{2019 IEEE International
  Symposium on Information Theory (ISIT)}.\hskip 1em plus 0.5em minus
  0.4em\relax IEEE, 2019, pp. 1107--1111.

\bibitem{DOliveira2019DegreeTF}
R.~G.~L. D'Oliveira, S.~{El Rouayheb}, D.~Heinlein, and D.~Karpuk, ``Degree
  tables for secure distributed matrix multiplication,'' in \emph{2019 IEEE
  Information Theory Workshop (ITW)}, 2019.

\bibitem{Aliasgari2019DistributedAP}
M.~Aliasgari, O.~Simeone, and J.~Kliewer, ``Distributed and private coded
  matrix computation with flexible communication load,'' \emph{2019 IEEE
  International Symposium on Information Theory (ISIT)}, pp. 1092--1096, 2019.

\bibitem{aliasgari2020private}
------, ``Private and secure distributed matrix multiplication with flexible
  communication load,'' \emph{IEEE Transactions on Information Forensics and
  Security}, vol.~15, pp. 2722--2734, 2020.

\bibitem{Kakar2019UplinkDownlinkTI}
J.~Kakar, A.~Khristoforov, S.~Ebadifar, and A.~Sezgin, ``Uplink-downlink
  tradeoff in secure distributed matrix multiplication,'' \emph{ArXiv}, vol.
  abs/1910.13849, 2019.

\bibitem{doliveira2020notes}
R.~G.~L. D’Oliveira, S.~E. Rouayheb, D.~Heinlein, and D.~Karpuk, ``Notes on
  communication and computation in secure distributed matrix multiplication,''
  in \emph{2020 IEEE Conference on Communications and Network Security (CNS)},
  2020, pp. 1--6.

\bibitem{Yu2020EntangledPC}
Q.~Yu and A.~S. Avestimehr, ``Entangled polynomial codes for secure, private,
  and batch distributed matrix multiplication: Breaking the "cubic" barrier,''
  \emph{ArXiv}, vol. abs/2001.05101, 2020.

\bibitem{mital2020secure}
N.~Mital, C.~Ling, and D.~Gündüz, ``Secure distributed matrix computation
  with discrete fourier transform,'' \emph{IEEE Transactions on Information
  Theory}, pp. 1--1, 2022.

\bibitem{bitar2021adaptive}
R.~Bitar, M.~Xhemrishi, and A.~Wachter-Zeh, ``Adaptive private distributed
  matrix multiplication,'' \emph{arXiv preprint arXiv:2101.05681}, 2021.

\bibitem{hasircioglu2021speeding}
B.~Hasircioglu, J.~Gomez-Vilardebo, and D.~Gunduz, ``Speeding up private
  distributed matrix multiplication via bivariate polynomial codes,''
  \emph{arXiv preprint arXiv:2102.08304}, 2021.

\bibitem{ftp-9606447}
R.~A. Machado, R.~G.~L. D’Oliveira, S.~E. Rouayheb, and D.~Heinlein, ``Field
  trace polynomial codes for secure distributed matrix multiplication,'' in
  \emph{2021 XVII International Symposium "Problems of Redundancy in
  Information and Control Systems" (REDUNDANCY)}, 2021, pp. 188--193.

\bibitem{9004505}
R.~G.~L. {D’Oliveira}, S.~{El Rouayheb}, and D.~{Karpuk}, ``Gasp codes for
  secure distributed matrix multiplication,'' \emph{IEEE Transactions on
  Information Theory}, pp. 1--1, 2020.

\bibitem{8949560}
Q.~Yu, M.~A. Maddah-Ali, and A.~S. Avestimehr, ``Straggler mitigation in
  distributed matrix multiplication: Fundamental limits and optimal coding,''
  \emph{IEEE Transactions on Information Theory}, vol.~66, no.~3, pp.
  1920--1933, 2020.

\bibitem{polycodes1}
Q.~Yu, M.~Maddah-Ali, and A.~S. Avestimehr, ``Polynomial codes: an optimal
  design for high-dimensional coded matrix multiplication,'' in \emph{Advances
  in Neural Information Processing Systems}, 2017, pp. 4403--4413.

\bibitem{polycodes2}
Q.~Yu, M.~A. Maddah-Ali, and A.~S. Avestimehr, ``Straggler mitigation in
  distributed matrix multiplication: Fundamental limits and optimal coding,''
  in \emph{2018 IEEE International Symposium on Information Theory
  (ISIT)}.\hskip 1em plus 0.5em minus 0.4em\relax IEEE, 2018, pp. 2022--2026.

\bibitem{pulkit}
S.~Dutta, M.~Fahim, F.~Haddadpour, H.~Jeong, V.~Cadambe, and P.~Grover, ``On
  the optimal recovery threshold of coded matrix multiplication,'' \emph{IEEE
  Transactions on Information Theory}, 2019.

\bibitem{pulkit2}
U.~Sheth, S.~Dutta, M.~Chaudhari, H.~Jeong, Y.~Yang, J.~Kohonen, T.~Roos, and
  P.~Grover, ``An application of storage-optimal matdot codes for coded matrix
  multiplication: Fast k-nearest neighbors estimation,'' in \emph{2018 IEEE
  International Conference on Big Data (Big Data)}.\hskip 1em plus 0.5em minus
  0.4em\relax IEEE, 2018, pp. 1113--1120.

\bibitem{fundamental}
S.~Li, M.~A. Maddah-Ali, Q.~Yu, and A.~S. Avestimehr, ``A fundamental tradeoff
  between computation and communication in distributed computing,'' \emph{IEEE
  Transactions on Information Theory}, vol.~64, no.~1, pp. 109--128, 2017.

\bibitem{nodehi2018limited}
H.~A. Nodehi and M.~A. Maddah-Ali, ``Limited-sharing multi-party computation
  for massive matrix operations,'' in \emph{2018 IEEE International Symposium
  on Information Theory (ISIT)}.\hskip 1em plus 0.5em minus 0.4em\relax IEEE,
  2018, pp. 1231--1235.

\bibitem{jia2019capacity}
Z.~Jia and S.~A. Jafar, ``On the capacity of secure distributed matrix
  multiplication,'' \emph{arXiv preprint arXiv:1908.06957}, 2019.

\bibitem{akbari2021secure}
H.~Akbari-Nodehi and M.~A. Maddah-Ali, ``Secure coded multi-party computation
  for massive matrix operations,'' \emph{IEEE Transactions on Information
  Theory}, vol.~67, no.~4, pp. 2379--2398, 2021.

\bibitem{kim2019private}
M.~Kim, H.~Yang, and J.~Lee, ``Private coded matrix multiplication,''
  \emph{IEEE Transactions on Information Forensics and Security}, vol.~15, pp.
  1434--1443, 2019.

\end{thebibliography}

\end{document}